\documentclass[12pt]{article}

\usepackage{mathtools}
\usepackage{amsthm}
\usepackage{eqnarray}
\usepackage{graphicx}
\usepackage{dcolumn}
\usepackage{bm}
\usepackage{verbatim}
\usepackage{subfigure}
\usepackage{xspace}
\usepackage{longtable}
\usepackage{array}
\usepackage{hyperref}
\usepackage{jabbrv}

\usepackage{tikz}

\usepackage{pgfornament}

\newtheorem{theorem}{Theorem}
\newtheorem{lemma}{Lemma}[theorem]
\newtheorem{corollary}[theorem]{Corollary}

\newcolumntype{L}{>{$}l<{$}} 
\newcolumntype{C}{>{$}c<{$}} 

\DefineSpuriousJournalWord{The}

\begin{document}

\title{Recurrence solution of monomer-polymer models on two-dimensional rectangular lattices}

\author{Yong Kong\\
  Department of Biostatistics \\
  School of Public Health \\
  Yale University \\
  New Haven, CT 06520, USA \\
  email: \texttt{yong.kong@yale.edu}
}%

\date{}
\maketitle

\begin{abstract}
  The problem of counting polymer coverings on the rectangular lattices is investigated.
  In this model, a linear rigid polymer covers $k$ 
  adjacent lattice sites
  such that no two polymers occupy a common site.
  Those unoccupied lattice sites are considered as monomers.
  We prove that for a given number of polymers ($k$-mers),
  the number of arrangements for the polymers on two-dimensional rectangular lattices
  satisfies simple recurrence relations.
  These recurrence relations are quite general and apply for arbitrary polymer length ($k$) and
  the width of the lattices ($n$).  The well-studied monomer-dimer problem is a special case
  of the monomer-polymer model
  when $k=2$.
  It is known the enumeration of monomer-dimer configurations in planar lattices is \#P-complete.
  The recurrence relations shown here have the potential for hints
  for the solution of
  long-standing problems in this class of computational complexity.
\end{abstract}


\section{Introduction}

The study of the combinatorial problem of enumerations of rigid rodlike molecules
on lattices has a long history and has attracted interests from researchers
in diverse fields of physics, mathematics, and theoretical computer science.
The system has been studied as a prototypical model for phase transitions
in equilibrium statistical mechanics
\cite{Onsager1949,floryPhaseEquilibriaSolutions1956,zwanzigFirstOrderPhase1963,Ghosh2007,Dhar2021,Rodrigues2023}.
In this model, a linear rigid polymer covers $k$ 
adjacent lattice sites
such that no two polymers occupy a common site.
The polymers are usually called $k$-mers,
and those unoccupied lattice sites are considered as monomers.

When $k=2$, the model becomes the well-studied monomer-dimer model.
In 1961 a special case of the monomer-dimer model, the close packed dimer problem 
where the lattice is fully covered by dimers, was solved 
analytically~\cite{Kasteleyn1961,Fisher1961,Temperley1961}.
It was further shown that counting dimer coverings of any \emph{planar} lattices
can be solved using the same Pfaffian method~\cite{Kasteleyn1963}.
For the general monomer-dimer problem where unoccupied lattice sites (monomers) are allowed,
no solution has been found despite decades of efforts
\cite{Baxter1968b,Gaunt1969,Heilmann1972,Samuel1980c,Hayn1994,Tzeng2003,kenyonDimersAmoebae2006,Wu2006d,Izmailian2005,Kong2006c,Kong2006,Kong2006b,Kong2007,Allegra2015,Pogudin2017}.
The importance of monomer-dimer problem also comes from the fact that
various other statistical physics problems can be mapped to the monomer-dimer model
\cite{Kasteleyn1963,Fisher1966b,McCoy1973}.  
Solution of the monomer-dimer problem will lead to the solution of these other problems.
For example,  the Ising model in the absence of an external field,
which was solved by Onsager in 1944~\cite{Onsager1944} using complicated methods,
is mapped
to the close packed dimer model,
while in the presence of an external field it is mapped
to the general monomer-dimer model.
The model also acts as the classical limit of the
recently introduced quantum dimer model 
which has been investigated intensively as the central model
in modern theories of strongly correlated quantum matter~\cite{RokhsarK88}.

The two-fold dichotomy of computational complexity of the monomer-dimer problem,
for counting dimer coverings of a \emph{planar} lattice and a \emph{nonplanar} lattice,
as well as for counting \emph{close packed} dimer configurations
and counting coverings \emph{with nonzero monomers},
has garnered the interest of
researchers in theoretical computer science
\cite{Jerrum1987,valiantComplexityComputingPermanent1979,valiantComplexityEnumerationReliability1979}.
It has been shown that the enumeration of monomer-dimer configurations in \emph{planar} lattices is \#P-complete~\cite{Jerrum1987},
which indicates the problem is computationally ``intractable''.
The class \#P plays the same role for counting problems 
(such as counting dimer configurations)
as the more familiar NP class does for decision problems
(such as the well-known satisfiability problem)
\cite{Garey1979,aroraComputationalComplexityModern2009,wigdersonMathematicsComputation2019,mooreNatureComputation2011,welshComplexityKnotsColourings1993}.
The \#P-complete problems are at least as hard as the NP-complete
problems in computational complexity hierarchy.
If \emph{any} problem in the \#P-complete class 
is found to be solvable, every problem in \#P class
is solvable.
Currently 
``P versus NP'' problem is perhaps the major
outstanding problem in theoretical computer science.

In this paper we give recurrence solutions of monomer-polymer models on two-dimensional rectangular lattices.
These recurrence relations are quite general and apply for arbitrary polymer length ($k$) and
the width of the lattices ($n$).
The paper is organized as the follows.
The major result is the Theorem~\ref{Th:recurrence} in Section~\ref{S:recur}.
The proof of the theorem is given in Section~\ref{S:proof}.
The methods used in the proof are elementary.
In Section~\ref{S:remarks} the implications of the results are discussed.

\section{Recurrence for the number of \texorpdfstring{$k$-mers}{k-mers} coverings} \label{S:recur}

Consider a $n \times m$ two-dimensional rectangular lattice
with $n$ lattice sites in the horizontal direction and $m$ lattice sites in the vertical direction.
In the horizontal direction two kinds of boundary conditions will be considered:
free boundary condition and cylinder boundary condition.  In the latter case
the lattice sites in the $n$th column are wrapped back and linked to the lattice sites in the first column.
In the vertical direction only free boundary conditions will be considered.

Let $s$ denote the number of $k$-mers in the lattice,
and denote the number of configurations of the $s$ $k$-mers on a lattice with a width of $n$ and a length of $m$
by $a(k, n,m,s)$. Since $k$, $n$ and $s$ are fixed, in the following $a(k, n,m,s)$ is
abbreviated as $a_{m,s}$ or $a_{m}$ for brevity.
In the following the notation $\binom{p}{q}$ is used for the binomial coefficient of $p$ choose $q$.

The main result is stated in the following theorem:
\begin{theorem}[Recurrence] \label{Th:recurrence}
For given $k$, $n$ and $s$, the following recursive relation holds:
\begin{equation} \label{E:rec1}
 \sum_{i=0}^{s} (-1)^i \binom{s}{i} a_{m-i, s} = c(n, k)^s, \qquad m \ge ks,
\end{equation}
where $c(n, k)$ is a constant that depends on the boundary conditions as well as $n$ and $k$,
but not $m$ or $s$.
For free boundary condition,
\[
  c(n, k) =
  \begin{cases}
    2n-k+1 , & n \ge k ,\\
    n   ,     & n < k  .
    \end{cases}
\]
For cylinder boundary condition,
\[
  c(n, k) =
  \begin{cases}
    2n ,   & n \ge k , \\
    n    ,   & n < k .
    \end{cases}
\]
\end{theorem}

From Theorem~\ref{Th:recurrence} we can obtain the following recurrence:
\begin{corollary}[Another recurrence]
\begin{equation} \label{E:rec2}
 \sum_{i=0}^{s+1} (-1)^i \binom{s+1}{i} a_{m-i, s} = 0, \qquad m \ge ks + 1
\end{equation}
\end{corollary}
\begin{proof}
Eq.~\eqref{E:rec2} can be derived from Eq.~\eqref{E:rec1}: 
substitute $m-1$ into $m$ in Eq.~\eqref{E:rec1}, and subtract the two equations:
\begin{align*}
  & a_m - \left[ \binom{s}{1} + \binom{s}{0} \right] a_{m-1}
    + \cdots
    + (-1)^s \left[  \binom{s}{s} + \binom{s}{s-1} \right] a_{m-s} - (-1)^{s} a_{m-1-s} \\
  =\enskip & a_m - \binom{s+1}{1} a_{m-1} + \binom{s+1}{2} - \cdots + (-1)^s \binom{s+1}{s} - (-1)^{s} a_{m-1-s} \\
  =\enskip & 0.
\end{align*}
In above the following binomial identity is used:
\begin{equation} \label{E:binom}
  \binom{j+1}{i} = \binom{j}{i} + \binom{j}{i-1}, \qquad i > 0 .
\end{equation}

\end{proof}

To illustrate the recurrence stated in Theorem~\ref{Th:recurrence},
some examples for different values of $k$ are listed in Table~\ref{T:a}.
In the table, $n=7$, $s=3$, and $m=20$. The lattices have free boundary conditions.
The $k$-mer lengths take the range of $2,3,4,5$.
The numbers are computed by extending the methods originally developed for the monomer-dimer problem ($k=2$) to handle arbitrary lengths of $k$-mers
\cite{kongGeneralRecurrenceTheory1999,Kong2006,Kong2006b,Kong2006c,Kong2007}.

For example,
when $k=3$, we have
\[
  1644154 - \binom{3}{1} 1383884 + \binom{3}{2} 1152702 - 948880 = 1728 = (2 \times 7 - 3 + 1)^3.
\]

\begin{table}
  \centering
  \caption{
The numbers $a_{m,s}$ for different $k$, with $n=7$, $s=3$, and $m=20$, on lattices with free boundary.
\label{T:a}}
  \begin{tabular}{LLLLL}
    \hline \hline
    &  \multicolumn{4}{C}{m} \\\cline{2-5}
   k & 17 & 18 & 19 & 20 \\
    \hline
    2 &  1491126  & 1788970 & 2124072 & 2498629 \\
    3 &  948880   & 1152702 & 1383884 & 1644154 \\
    4 &  560552   & 692133  & 843008  & 1014508 \\
    5 &  305326   & 384872  & 477398  & 583904  \\
\hline
\end{tabular}
\end{table}

\section{Proof of the recurrence} \label{S:proof}

First we define the possible states of a $k$-mer
covering a given lattice site.
Notice that for a given lattice site, there are $k+2$ configurational states for a $k$-mer.
The monomer state, where the lattice site is empty, is denoted as state $0$;
When the lattice site is occupied by the first part of the $k$-mer (in the positive $m$ direction),
we say the $k$-mer is in state $1$ for the given lattice site,
and so on.
When the last part of the $k$-mer occupies the lattice site, it's in state $k$.
If the lattice site is occupied by a horizontal $k$-mer, a state $k+1$ is assigned.
Figure~\ref{F:states} shows the $5$ states for a trimer ($k=3$).
Note that the states are specific for a lattice site in a given row:
if the $k$-mer is in a state $j$ for a lattice site in the $i$th row,
then it is in state $j-1$ for the site of the $(i+1)$th row,

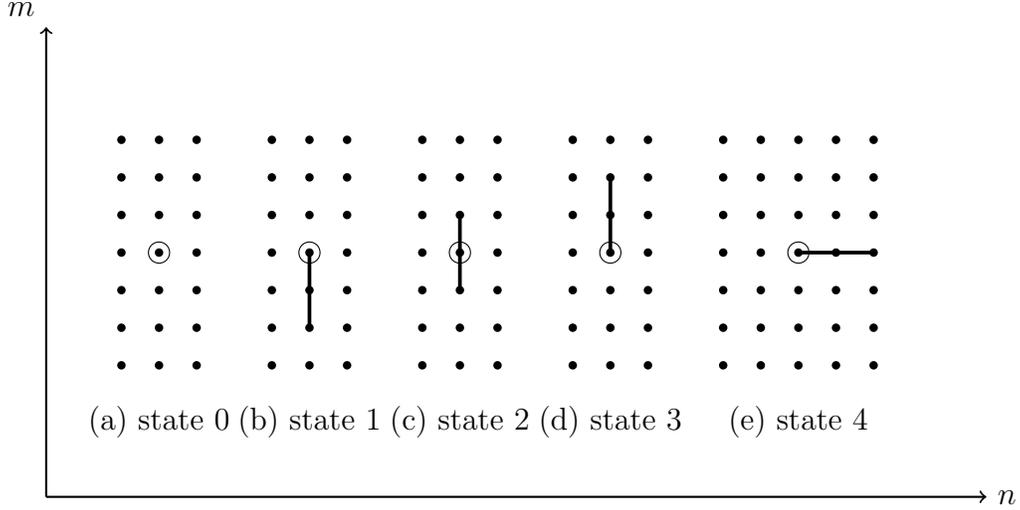
\begin{figure}[!ht]
\begin{tikzpicture}[x=5mm, y=5mm]

  \foreach \x in {0,1,2, 4,5,6, 8,9,10, 12,13,14, 16,17,18,19,20}
    \foreach \y in {0,1,...,6}
       \fill (\x, \y) circle[radius=1.6pt];

\draw (1,3) circle (4pt);
\draw (5,3) circle (4pt);
\draw (9,3) circle (4pt);
\draw (13,3) circle (4pt);
\draw (18,3) circle (4pt);

\draw[line width=1.4pt] (5,1) -- (5,3);
\draw[line width=1.4pt] (9,2) -- (9,4);
\draw[line width=1.4pt] (13,3) -- (13,5);
\draw[line width=1.4pt] (18,3) -- (20,3);

\draw[thick,->] (-2,-3.5) -- (23,-3.5) node[anchor=west] {$n$};
\draw[thick,->] (-2,-3.5) -- (-2, 9) node[anchor=south east] {$m$};

\node at (1, -1.5)   (0) {(a) state $0$};
\node at (5, -1.5)   (1) {(b) state $1$};
\node at (9, -1.5)   (2) {(c) state $2$};
\node at (13, -1.5)  (3) {(d) state $3$};
\node at (18, -1.5)  (4) {(e) state $4$};

\end{tikzpicture}
\caption{The configurational states of one $k$-mer on one lattice site.
  For a given lattice site each $k$-mer can have $k+2$ states.
  Here is an example for $k=3$.
  For the circled center site,
  the trimer can have $5$ states:
  (a) state $0$: the site is empty;
  (b) state $1$: the site is occupied by the first part of a vertical trimer;
  (c) state $2$: the site is occupied by the middle part a vertical trimer;
  (d) state $3$: the site is occupied by the end part a vertical trimer;
  (e) state $4$: the site is occupied by a horizontal trimer.
}\label{F:states}
\end{figure}

For proving the recurrence, we need to define certain numbers
that count $k$-mer configurations with some $k$-mers restricted to given rows.
   Let $a_{(p_1, p_2, \cdots)m,s}^{(q_1, q_2. \cdots)}$ denote the number of configurations to have $s$ $k$-mers
   in a $n \times m$ lattice with the requirement that
   at least one $k$-mer touching the $p_1$th row from above (and none from below),
   at least one $k$-mer touching the $p_2$th row from above (and none from below), etc.,
   and
   at least one $k$-mer touching the $q_1$th row from below (and none from above),
   at least one $k$-mer touching the $q_2$th row from below (and none from above), etc..
   For those $k$-mers restricted in the \emph{subscript} of this notation,
   the $k$-mers are in the states of $k$ or $k+1$ with respect to the $p_i$th row,
   while those $k$-mers restricted in the \emph{superscript} 
   are in the states of $1$ or $k+1$ with respect to the $q_i$th row.

   Furthermore, another notation will be introduced: if the numbers $p_i$ in the subscripts have a bar, it means that
   there is at least one $k$-mers with state in the set $\{1, 2, \dots, k-1\}$
   with respect of the $p_i$th row of the lattice: $a_{(p_1, p_2, \cdots, \overline{p_i}, \cdots)m,s}^{(q_1, q_2. \cdots)}$.
   
   In the proof below, we only need $p_i$'s to be in the rows of the set
   $\{1, k+1, 2k+1, \dots, \}$,
   and $q_i$'s to be the last row, i.e., the $m$th row.
   Hence the notation under these conditions becomes
   $a_{(1, k+1, 2k+1, \cdots)m,s}^{(m)}$.
   For brevity, a shorthand version will be used, which only uses the coefficients of $k$ in the \emph{subscripts}:
   \[
     a_{(0, 1, 2, \cdots)m,s}^{(m)} ,
   \]
   and
   \[
     a_{(0, 1, 2, \cdots, \bar{i}, \cdots)m,s}^{(m)} .
   \]

   The sum of the numbers of rows in the subscript and superscript within the parentheses is the number of restricted
   $k$-mers. The rest of the $k$-mers can be arranged freely on the lattice
   (but not in contradiction to those restricted rows).
   Note that by symmetry, the superscript and subscript can be switched in totality without affecting the value.
   With these definitions, we first prove the following lemmas.
   The proof of the recurrence in Theorem~\ref{Th:recurrence} is based on the following three lemmas.

   The first lemma converts the difference between the enumerations of
   a lattice with a length $m$ and a lattice with a length $m-1$
   to an enumeration with restrictions on the last row.%
   \begin{lemma}[Go to the top]  \label{L:top}
     \begin{equation}            \label{E:top}
       a_{(0, 1, 2, \cdots, t)m,s} - a_{(0, 1, 2, \cdots, t)m-1,s} = a_{(0, 1, 2, \cdots, t)m,s}^{(m)}, 
     \end{equation}
     where the coefficients in the subscripts within the parentheses of the three terms are the same,
     and can have bars over them.
   \end{lemma}
   \begin{proof}
     Among all the $k$-mers configurations counted by $a_{(0, 1, 2, \cdots)m,s}$,
     some configurations have one or more $k$-mers
     touching the last row.
     These configurations are counted by $a_{(0, 1, 2, \cdots)m,s}^{(m)}$.
     For the rest of configurations
     none of the $k$-mers cover any sites of the last row: the last row is empty (unoccupied by $k$-mers).
     This is equivalent to a lattice of length $m-1$, hence
     the number of these configurations is counted by $a_{(0, 1, 2, \cdots)m-1,s}$.
     \end{proof}

     The second lemma converts an enumeration with restrictions on the last row (the $m$th row)
     to a sum of two terms:
     one term with the restrictions on the last row moved
     to the restrictions on a row at the bottom of the lattice
     that are $k$ sites above the topmost previously restricted row (the $[(t+1)k+1]$th row),
     and the other term is with the barred version on the same row, but with the restriction on the last row
     unchanged.
   \begin{lemma}[Go to the bottom] \label{L:bottom}
     \begin{align}                 \label{E:bottom}
       a_{(0,\dots,t)m,s}^{(m)}     &= a_{(0,\dots,t,t+1)m,s} + a_{(0,\dots,t,\overline{t+1})m,s}^{(m)}  , \\
       a_{(0,\dots,\bar{t})m,s}^{(m)} &= a_{(0,\dots,\bar{t},t+1)m,s} + a_{(0,\dots,\bar{t},\overline{t+1})m,s}^{(m)}  .
     \end{align}
   \end{lemma}
   \begin{proof}
     For these $k$-mers configurations counted by $a_{(0,\dots,t)m,s}^{(m)}$,
     consider the $((t+1)k+1)$th row.
     The row can be empty,
     or if there are $k$-mers that intersect with this row,
     either they only have states in the set of $\{k, k+1\}$ on this row,
     or there exists at least one $k$-mer
     that is of state $\{1,\dots,k-1\}$.
     For the former case (empty or in the states of $\{k, k+1\}$), we can \emph{flip} the lattice between the $((t+1)k+1)$th row and the $m$th row,
     and the number of configurations in the flipped lattice are counted as $a_{(0,\dots,t,t+1)m,s}$.
     The rest of the configurations are counted by $a_{(0,\dots,t,\overline{t+1})m,s}^{(m)}$.
     Similar arguments can be used for the second equation.
     \end{proof}

In the following we use $a_{(0, \dots, \bar{t})  m}^{(m)}(j)$ to emphasize
that there are
$j$ $k$-mers that
are not restricted to the rows specified by the subscripts or the superscript,
with $t+2+j=s$.
\begin{lemma}[$j$ unrestricted $k$-mers] \label{L:j-free}
\begin{align} 
  \sum_{i=0}^{j+1} (-1)^i \binom{j+1}{i} a_{(0, \dots, \bar{t})  m-i, s}^{(m-i)}(j) &= 0, \label{E:jfree-1} \\
  \sum_{i=0}^{j+1} (-1)^i \binom{j+1}{i} a_{(0, \dots, {t})  m-i, s}^{(m-i)}(j) &= 0, \label{E:jfree-2}
\end{align}
where $j$ stands for the number of unrestricted $k$-mers, with $t+2+j=s$, and $m \ge ks + 1$.
Similar identities also hold for
$a_{(0, \dots, \bar{t})  m-i, s}$
and
$a_{(0, \dots, {t})  m-i, s}$,
where there is no restriction on the last row.
In these cases $t+1+j=s$.
\end{lemma}

\begin{proof}

These identities can be proved by mathematical induction.
For $j=0$ the identities are true trivially:
all $s$ $k$-mers are constrained in the first $s-1$ rows and the last row,
with one and only one $k$-mer intersect with each row, 
hence
\begin{align*}
  a_{(0, \dots, \bar{t})  m}^{(m)}(0) - a_{(0, \dots, \bar{t})  m-1}^{(m-1)}(0) &= 0, \\
  a_{(0, \dots, {t})  m}^{(m)}(0)    - a_{(0, \dots, {t})  m-1}^{(m-1)}(0)  &= 0 .
\end{align*}

Assume the identities are true for $j-1$, then for $j$ unrestricted $k$-mers we can use Lemma~\ref{L:top} and Lemma~\ref{L:bottom}:
\begin{align*}
  & \sum_{i=0}^{j+1} (-1)^i \binom{j+1}{i} a_{(0, \dots, \bar{t})  m-i}^{(m-i)}(j)  \\
 = & \sum_{i=0}^{j} (-1)^i \binom{j}{i}
    \left[ a_{(0, \dots, \bar{t})  m-i}^{(m-i)}(j) -  a_{(0, \dots, \bar{t})  m-i-1}^{(m-i-1)}(j) \right]\\
 = & \sum_{i=0}^{j} (-1)^i \binom{j}{i}
    \left[ [ a_{(0,\dots,\bar{t},t+1)m-i}(j) + a_{(0, \dots, \bar{t},\overline{t+1})  m-i}^{(m-i)}(j-1) ] \right. \\
    &\left.{} - [ a_{(0,\dots,\bar{t},t+1)m-i-1}(j) + a_{(0, \dots, \bar{t},\overline{t+1})  m-i-1}^{(m-i-1)}(j-1) ] \right] \\
  = & \sum_{i=0}^{j} (-1)^i \binom{j}{i}
    \left[ a_{(0,\dots,\bar{t},t+1)m-i}^{(m-i)}(j-1)
    +  a_{(0, \dots, \bar{t},\overline{t+1})  m-i}^{(m-i)}(j-1) - a_{(0, \dots, \bar{t},\overline{t+1})  m-i-1}^{(m-i-1)}(j-1) \right] .
\end{align*}
In the first step the binomial identity Eq.~\eqref{E:binom} is used.
In the second step Lemma~\ref{L:bottom} is used,
and in the third step Lemma~\ref{L:top} is used.
The three summands in the last sum all have $j-1$ unrestricted $k$-mers,
and by the assumption their sums all vanish.

The second identity Eq.~\eqref{E:jfree-2} and the identities without superscripts
can be proved similarly.

\end{proof}

Note that when $j=s$, Eq.~\eqref{E:rec2} is obtained by using the identities of Lemma~\ref{L:j-free} without superscripts.

Finally, the recurrence of Eq.~\eqref{E:rec1} in Theorem~\ref{Th:recurrence} can be proved.
In the proof the following steps are used repetitively, reducing the number of unrestricted $k$-mers by one in each iteration:
the binomial identity Eq.~\eqref{E:binom} is used
to write the summand into a difference of two terms (Eq.~\eqref{E:s-1} and Eq.~\eqref{E:s-2}),
Lemma~\ref{L:top} is then used to convert the difference of the two terms
into one term by introducing the restriction to the top row (Eq.~\eqref{E:s-1b} and Eq.~\eqref{E:s-2b},
where in Eq.~\eqref{E:s-1b} $a_{(0)m-i} = a_{m-i}^{(m-i)}$),
thus reducing the number of unrestricted $k$-mers by one.
The Lemma~\ref{L:bottom} is then used to bring this restriction
on the top row to the bottom (Eq.~\eqref{E:s-2c}),
releasing the constraints on the top row so that Lemma~\ref{L:top} can be applied again in the next iteration.
In doing so this step
introduces an extra term with a bar in the subscripts.
Then Lemma~\ref{L:j-free} is used to make the sum of this extra term vanish (Eq.~\eqref{E:s-2d}).
After one iteration of these steps,
the number of unrestricted $k$-mers is reduced by one.
In the final step, only one term remains,
where the restrictions are 
on bottom $s$ rows of the lattice, with the rows indexed as: $1, k+1, 2k+1, \dots, (s-1)k+1$, each separated from the next by a length
of $k$. Thus the $s$ $k$-mers can be arranged on these $s$ rows independently, leading to the final result.
\begin{proof}[Proof of Theorem~\ref{Th:recurrence}]
  \begin{align}
    &\sum_{i=0}^{s} (-1)^i \binom{s}{i} a_{m-i}    \\
    =& \sum_{i=0}^{s-1} (-1)^i \binom{s-1}{i} \left[ a_{m-i} - a_{m-i-1} \right] \label{E:s-1} \qquad \text{by Eq.~\eqref{E:binom}}  \\ 
    =& \sum_{i=0}^{s-1} (-1)^i \binom{s-1}{i}  a_{(0)m-i} \qquad \text{by Lemma~\ref{L:top} } \label{E:s-1b}\\
    =& \sum_{i=0}^{s-2} (-1)^i \binom{s-2}{i} \left[  a_{(0)m-i} - a_{(0)m-i-1} \right]  \qquad \text{by Eq.~\eqref{E:binom}}  \label{E:s-2} \\
    =& \sum_{i=0}^{s-2} (-1)^i \binom{s-2}{i}  a_{(0)m-i}^{(m-i)}  \qquad \text{by Lemma~\ref{L:top} }  \label{E:s-2b} \\
    =& \sum_{i=0}^{s-2} (-1)^i \binom{s-2}{i} \left[ a_{(0, 1)m-i} + a_{(0, \bar{1})m-i}^{(m-i)} \right] \qquad \text{by Lemma~\ref{L:bottom} }  \label{E:s-2c}\\
    =& \sum_{i=0}^{s-2} (-1)^i \binom{s-2}{i} a_{(0, 1)m-i}  \qquad \text{by Lemma~\ref{L:j-free} } \label{E:s-2d} \\
    =& \cdots \\
    =& a_{(0, 1, \dots, s-1)m} = c(n, k)^s.
  \end{align}
  In the last step, the $s$ $k$-mers are in $s$ rows of the lattice that are $k$ sites apart, so their configurations
  are independent of each other.
  The number of arrangements for one $k$-mer in a lattice strip of $n \times k$ is denoted by $c(n, k)$.
  For the free boundary condition,
  when $n \ge k$,
  \[
    c(k, n) = (n-k+1) + n = 2n - k + 1,
  \]
  where the terms within parenthesis is for the $k$-mer in horizontal orientation,
  and the last term $n$ is for the $k$-mer in vertical orientation.
  When $n < k$,
  the first term does not exist and $c(k, n) = n$.

  For the cylinder boundary condition,
  when $n \ge k$,
  there are $n$ configurations for the $k$-mer in horizontal orientation,
  and $n$ configurations in vertical orientation, so
   $c(k, n) = 2n$.
  When $n < k$, $c(k, n) = n$ for vertical orientation only, as in the case of the free boundary condition.  
\end{proof}

\section{Concluding remarks} \label{S:remarks}

The recurrences of Eq.~\eqref{E:rec1} and Eq.~\eqref{E:rec2} show
that the enumerations of monomer-polymer coverings in two-dimensional lattices
have some regular patterns.
From the recurrences we can obtain the generating function for the number of coverings for a lattice with width $n$
and a fixed number of $k$-mers of $s$
in the form of~\cite{stanley_2011}
\[
   G(x; k,n,s) = \sum_{m=0}^\infty a_{m,s} x^m =  \frac{ c(n, k)^s x^{ks} } { (1-x)^{s+1} } + \frac{ P(x; k,n,s) } {(1-x)^s },
\]
where $P(x; k,n,s)$ is a polynomial in the indeterminate variable $x$ with a degree of $ks -1$.
The polynomial $P(x; k,n,s)$ is determined by the initial conditions of the recurrences.
The generating function has only one pole at $x=1$ with a multiplicity of $s+1$.
The expansion of the first term by the binomial theorem gives
\[
  c(n, k)^s \binom{m+s-ks}{s}
\]
as the contribution to $a_{m,s}$ from the first team,
but the expansion of the second term depends on the polynomial $P(x; k,n,s)$.
It would be interesting to find out what kinds of patterns the polynomial $P(x; k,n,s)$ has
and how they contribute to the enumerations of monomer-polymer coverings.

\bibliographystyle{jabbrv_abbrv} 
\bibliography{kmer_full_jname_2023_09_17.bib}

\begin{thebibliography}{10}

\bibitem{Allegra2015}
N.~Allegra.
\newblock Exact solution of the 2d dimer model: Corner free energy, correlation
  functions and combinatorics.
\newblock {\em\JournalTitle{Nuclear Physics B}}, 894:685--732, 2015.

\bibitem{aroraComputationalComplexityModern2009}
S.~Arora and B.~Barak.
\newblock {\em Computational {{Complexity}}: {{A Modern Approach}}}.
\newblock Cambridge University Press, 1 edition, Apr. 2009.

\bibitem{Baxter1968b}
R.~J. Baxter.
\newblock Dimers on a rectangular lattice.
\newblock {\em\JournalTitle{Journal of Mathematical Physics}}, 9(4):650--654,
  1968.

\bibitem{Dhar2021}
D.~Dhar and R.~Rajesh.
\newblock Entropy of fully packed hard rigid rods on d-dimensional hypercubic
  lattices.
\newblock {\em\JournalTitle{Physical Review E}}, 103(4), 2021.

\bibitem{Fisher1961}
M.~E. Fisher.
\newblock Statistical mechanics of dimers on a plane lattice.
\newblock {\em\JournalTitle{Physical Review}}, 124(6):1664--1672, 1961.

\bibitem{Fisher1966b}
M.~E. Fisher.
\newblock On the dimer solution of planar {I}sing models.
\newblock {\em\JournalTitle{Journal of Mathematical Physics}},
  7(10):1776--1781, 1966.

\bibitem{floryPhaseEquilibriaSolutions1956}
P.~J. Flory.
\newblock Phase equilibria in solutions of rod-like particles.
\newblock {\em\JournalTitle{Proceedings of the Royal Society of London. Series
  A. Mathematical and Physical Sciences}}, 234(1196):73--89, 1956.

\bibitem{Garey1979}
M.~R. Garey and D.~S. Johnson.
\newblock {\em Computers and Intractability, A Guide to the Theory of
  {NP}-Completeness}.
\newblock W.H. Freeman and Company, New York, 1979.

\bibitem{Gaunt1969}
D.~Gaunt.
\newblock Exact series-expansion study of the monomer-dimer problem.
\newblock {\em\JournalTitle{Phys. Rev.}}, 179:174--186, 1969.

\bibitem{Ghosh2007}
A.~Ghosh and D.~Dhar.
\newblock On the orientational ordering of long rods on a lattice.
\newblock {\em\JournalTitle{Europhysics Letters}}, 78(2), 2007.

\bibitem{Hayn1994}
R.~Hayn and V.~N. Plechko.
\newblock Grassmann variable analysis for dimer problems in 2 dimensions.
\newblock {\em\JournalTitle{Journal of Physics a-Mathematical and General}},
  27(14):4753--4760, 1994.

\bibitem{Heilmann1972}
O.~J. Heilmann and E.~H. Lieb.
\newblock Theory of monomer-dimer systems.
\newblock {\em\JournalTitle{Communications in Mathematical Physics}},
  25(3):190--232, 1972.

\bibitem{Izmailian2005}
N.~Izmailian, V.~B. Priezzhev, P.~Ruelle, and C.~K. Hu.
\newblock Logarithmic conformal field theory and boundary effects in the dimer
  model.
\newblock {\em\JournalTitle{Phys Rev Lett}}, 95(26):260602, 2005.

\bibitem{Jerrum1987}
M.~Jerrum.
\newblock Two-dimensional monomer dimer systems are computationally
  intractable.
\newblock {\em\JournalTitle{Journal of Statistical Physics}}, 48(1-2):121--134,
  1987.
\newblock Erratum: 59, 1087-1088 (1990).

\bibitem{Kasteleyn1961}
P.~W. Kasteleyn.
\newblock The statistics of dimers on a lattice.
\newblock {\em\JournalTitle{Physica}}, 27(12):1209--1225, 1961.

\bibitem{Kasteleyn1963}
P.~W. Kasteleyn.
\newblock Dimer statistics and phase transitions.
\newblock {\em\JournalTitle{Journal of Mathematical Physics}}, 4(2):287--293,
  1963.

\bibitem{kenyonDimersAmoebae2006}
R.~Kenyon, A.~Okounkov, and S.~Sheffield.
\newblock Dimers and amoebae.
\newblock {\em\JournalTitle{Annals of Mathematics}}, 163(3):1019--1056, May
  2006.

\bibitem{kongGeneralRecurrenceTheory1999}
Y.~Kong.
\newblock General recurrence theory of ligand binding on a three-dimensional
  lattice.
\newblock {\em\JournalTitle{The Journal of Chemical Physics}},
  111(10):4790--4799, Sept. 1999.

\bibitem{Kong2006}
Y.~Kong.
\newblock Logarithmic corrections in the free energy of monomer-dimer model on
  plane lattices with free boundaries.
\newblock {\em\JournalTitle{Phys. Rev. E}}, 74:011102, Jul 2006.

\bibitem{Kong2006b}
Y.~Kong.
\newblock Monomer-dimer model in two-dimensional rectangular lattices with
  fixed dimer density.
\newblock {\em\JournalTitle{Phys. Rev. E}}, 74:061102, Dec 2006.

\bibitem{Kong2006c}
Y.~Kong.
\newblock Packing dimers on $(2p+1)\ifmmode\times\else\texttimes\fi{}(2q+1)$
  lattices.
\newblock {\em\JournalTitle{Phys. Rev. E}}, 73:016106, Jan 2006.

\bibitem{Kong2007}
Y.~Kong.
\newblock Asymptotics of the monomer-dimer model on two-dimensional
  semi-infinite lattices.
\newblock {\em\JournalTitle{Phys. Rev. E}}, 75:051123, May 2007.

\bibitem{McCoy1973}
B.~M. McCoy and T.~T. Wu.
\newblock {\em The Two-Dimensional {{Ising Model}}}.
\newblock Harvard University Press, Cambridge, Massachusetts, 1973.

\bibitem{mooreNatureComputation2011}
C.~Moore and S.~Mertens.
\newblock {\em {The Nature of Computation}}.
\newblock Oxford University Press, 2011.

\bibitem{Onsager1944}
L.~Onsager.
\newblock Crystal statistics. {I}. {A} two-dimensional model with an
  order-disorder transition.
\newblock {\em\JournalTitle{Physical Review}}, 65(3/4):117--149, 1944.

\bibitem{Onsager1949}
L.~Onsager.
\newblock The effects of shape on the interaction of colloidal particles.
\newblock {\em\JournalTitle{Annals of the New York Academy of Sciences}},
  51(4):627--659, 1949.

\bibitem{Pogudin2017}
G.~Pogudin.
\newblock Power series expansions for the planar monomer-dimer problem.
\newblock {\em\JournalTitle{Physical Review E}}, 96(3), 2017.

\bibitem{Rodrigues2023}
L.~R. Rodrigues, J.~F. Stilck, and W.~G. Dantas.
\newblock Entropy of rigid k-mers on a square lattice.
\newblock {\em\JournalTitle{Physical Review E}}, 107(1), 2023.

\bibitem{RokhsarK88}
D.~S. Rokhsar and S.~A. Kivelson.
\newblock Superconductivity and the quantum hard-core dimer gas.
\newblock {\em\JournalTitle{Phys. Rev. Lett.}}, 61:2376--2379, 1988.

\bibitem{Samuel1980c}
S.~Samuel.
\newblock The use of anticommuting variable integrals in statistical mechanics.
  {III}. {Unsolved} models.
\newblock {\em\JournalTitle{J. Math. Phys.}}, 21(12):2820--2833, Dec. 1980.

\bibitem{stanley_2011}
R.~P. Stanley.
\newblock {\em Enumerative Combinatorics}, volume~1 of {\em Cambridge Studies
  in Advanced Mathematics}.
\newblock Cambridge University Press, 2 edition, 2011.

\bibitem{Temperley1961}
H.~N.~V. Temperley and M.~E. Fisher.
\newblock Dimer problem in statistical mechanics - an exact result.
\newblock {\em\JournalTitle{Philosophical Magazine}}, 6(68):1061--1063, 1961.

\bibitem{Tzeng2003}
W.~J. Tzeng and F.~Y. Wu.
\newblock Dimers on a simple-quartic net with a vacancy.
\newblock {\em\JournalTitle{Journal of Statistical Physics}},
  110(3-6):671--689, 2003.

\bibitem{valiantComplexityComputingPermanent1979}
L.~G. Valiant.
\newblock The complexity of computing the permanent.
\newblock {\em\JournalTitle{Theoretical Computer Science}}, 8(2):189--201,
  1979.

\bibitem{valiantComplexityEnumerationReliability1979}
L.~G. Valiant.
\newblock The complexity of enumeration and reliability problems.
\newblock {\em\JournalTitle{SIAM Journal on Computing}}, 8(3):410--421, Aug.
  1979.

\bibitem{welshComplexityKnotsColourings1993}
D.~J.~A. Welsh.
\newblock {\em Complexity: Knots, Colourings, and Counting}.
\newblock Number 186 in London Mathematical Society Lecture Note Series.
  Cambridge University Press, Cambridge ; New York, 1993.

\bibitem{wigdersonMathematicsComputation2019}
A.~Wigderson.
\newblock {\em Mathematics and Computation}.
\newblock Princeton University Press, Princeton, NJ, 2019.

\bibitem{Wu2006d}
F.~Y. Wu.
\newblock Pfaffian solution of a dimer-monomer problem: Single monomer on the
  boundary.
\newblock {\em\JournalTitle{Physical Review E}}, 74(2 Pt 1):020104, 2006.

\bibitem{zwanzigFirstOrderPhase1963}
R.~Zwanzig.
\newblock First-{{Order Phase Transition}} in a {{Gas}} of {{Long Thin Rods}}.
\newblock {\em\JournalTitle{The Journal of Chemical Physics}},
  39(7):1714--1721, Oct. 1963.

\end{thebibliography}

\end{document}